\documentclass[11pt, english]{article}
\usepackage{babel}
\usepackage{amsmath, amssymb, amsthm, eucal}
\usepackage{graphicx}
\usepackage[shortlabels]{enumitem}

\usepackage[textheight=630pt,
  textwidth=468pt,
  centering]{geometry}

\usepackage{color}

\newtheorem{theorem}{Theorem}

\newtheorem{lemma}{Lemma}

\theoremstyle{definition}

\newcommand{\eps}{\varepsilon}

\newcommand{\abs}[1]{\left| #1 \right|}

\newcommand{\norm}[1]{\Vert #1 \Vert}

\newcommand{\RR}{\mathbb{R}}

\newcommand{\CC}{\mathbb{C}}
\newcommand{\PP}{\mathbb{P}}

\newcommand{\brac}[1]{\left(#1\right)}
\newcommand{\fbrac}[1]{\left\{#1\right\}}
\newcommand{\sbrac}[1]{\left[#1\right]}
\newcommand{\abrac}[1]{\langle#1\rangle}

\DeclareMathOperator{\mean}{\mathbb{E}}
\DeclareMathOperator{\supp}{supp}

\DeclareMathOperator{\sgn}{sgn}
\DeclareMathOperator{\cone}{cone}

\DeclareMathOperator{\sph}{\mathbb{S}}
\DeclareMathOperator{\dist}{dist}

\begin{document}

\title{Robust analysis $\ell_1$-recovery from Gaussian measurements 
and total variation minimization}
\date{}
\author{M.~Kabanava\thanks{RWTH Aachen University, Chair for Mathematics C (Analysis), Templergraben 55, 52062 Aachen Germany, \tt{kabanava@mathc.rwth-aachen.de}},  H.~Rauhut\thanks{RWTH Aachen University, Chair for Mathematics C (Analysis), Templergraben 55, 52062 Aachen Germany, \tt{rauhut@mathc.rwth-aachen.de}}, H.~Zhang\thanks{College of Science, National University of Defense Technology, Changsha, Hunan, 410073 China \tt{hhuuii.zhang@gmail.com}}}

\maketitle

\begin{abstract}
Analysis $\ell_1$-recovery refers to a technique of recovering a signal that is sparse in some transform domain from incomplete corrupted measurements.
This includes total variation minimization as an important special case when the transform domain is generated by a difference operator. In the present paper we provide a bound on the number of Gaussian measurements required for successful recovery
for total variation and for the case that the analysis operator is a frame. The bounds are particularly suitable when the sparsity of the analysis representation of the signal is not very small.
\end{abstract}


\section{Introduction}

Compressive sensing is a recent field in signal processing that predicts that sparse vectors can be stably reconstructed from incomplete measurements via efficient algorithms \cite{fora13,elku12}.
Traditionally, the synthesis sparsity model is used in this context, where it is assumed that the signal can be written as a linear combination of a few elements of an appropriate basis.
Recently, the analysis sparsity model (or cosparsity model) has attracted significant interest as well \cite{caelnera11,daelgrna11,daelgigrna14,kara13}, where it is assumed that the signal is sparse after a transformation.
In the case of a basis transformation, the synthesis and analysis models coincide, but if the transform is redundant then the analysis sparsity model is different, and in fact the class of analysis sparsity
models is richer than the class of synthesis sparsity models, see \cite{EladMilanfarRubinstein,daelgrna11} for further details.
While the naive recovery approach of $\ell_0$-minimization is NP-hard, one may use convex relaxations, i.e., $\ell_1$-minimization in both the synthesis and analysis sparsity cases. Especially in the synthesis sparsity case, $\ell_1$-minimization is by now rather well understood, see e.g.\ \cite{fora13}. Despite recent progress in the theory of
analysis $\ell_1$-minimization \cite{caelnera11,daelgrna11,daelgigrna14,kara13},
there still remain a number of questions to be explored. In particular, in the important special case of total variation minimization which is ubiquitious in image processing
only a few contributions analyzing bounds for recovery from underdetermined measurements are available \cite{newa12,NeedellWard13,CaiXu}. The results of \cite{newa12,NeedellWard13} cover stable and robust recovery of signals in $\CC^{d^n}$ of arbitrary dimension $n\geq 2$. They rely on the restricted isometry property and bounds on the decay of wavelet coefficients of high-dimensional signals by its total variation semi-norm. The authors of \cite{CaiXu} provide a bound on the number of Gaussian measurements that guarantee stable and robust recovery of signals in $\RR^{d^n}$ with $n\geq 1$. The recovery is proved by establishing the null space property for the measurement matrix. For $n=1$ and fixed signal dimension $d$ bounds on the number of measurements were computed using Gordon's escape through a mesh theorem. In the asymptotic regime when $d\to\infty$ the performance guarantees explore the Grassmann angle framework. This article contributes to the topic of total variation minimization by providing a bound on the number of Gaussian random measurements in order to recover a gradient sparse signal via total variation minimization.
Moreover, we also provide an alternative bound to the one in \cite{kara13}  on the number
of required Gaussian measurements for recovery of analysis-sparse signals with respect to a frame.

In contrast to \cite{newa12,NeedellWard13,CaiXu} which establish uniform recovery of all signals simultaneously with a single draw of a measurement matrix, our main results concern so-called nonuniform recovery using a Gaussian random measurement matrix, i.e., we fix a sparse (or rather cosparse)
vector (with respect to a given analysis operator) and provide bounds that guarantee that the given vector is recovered via analysis $\ell_1$-minimization
with high probability. Our bounds are particularly good when the sparsity is not very small compared to the ambient dimension. In fact, in the case of analysis-sparsity with respect to a very redundant frame,
there are natural lower bounds for the sparsity, so that in this context
the bounds in \cite{kara13} may turn out to be trivial in certain situations as they require more measurements than the ambient dimension. In contrast, the bounds
derived in this paper are always non-trivial in the sense that the number of required measurements is always lower than the ambient dimension.

Following \cite{chparewi12}, our analysis is based on estimating widths of tangent cones of a transformed $\ell_1$-norm at the (co-)sparse vector to be recovered, similarly
to \cite{kara13} or \cite[Chapter 9.2]{fora13}.
This is in contrast to uniform recovery bounds which are often based on the restricted isometry property or the null space property \cite{fora13}, see also \cite{kara13} for precise bounds
for an analysis sparsity version of the null space property for Gaussian random measurements.

%

In mathematical terms, we wish to recover a signal $x\in\RR^d$ from measurements
\begin{equation}\label{eqMeasurements}
y=Mx+w,
\end{equation}
where $M\in\RR^{m\times d}$ is a measurement matrix and $w \in \RR^m$ with $\norm{w}_2\leq\eta$ corresponds to noise.
When $m\ll d$, there are infinitely many solutions to (\ref{eqMeasurements}). However, the prior sparsity knowledge about the underlying signal $x$ makes its recovery possible.

We assume that $x$ possesses  a structure generated by a matrix $\Omega \in \RR^{p \times d}$, called the analysis operator, in the sense that the application of $\Omega$ to $x$ produces a vector
with a small number of non-zero entries.  If $\Omega x$ has $s$ non-zero entries, then $x$ is called $\ell$-\emph{cosparse}, where the number $\ell:=p-s$ is refered to as \emph{cosparsity} of $x$
(with respect to $\Omega$).
The index set of the zero entries of $\Omega x$ is called the \emph{cosupport} of $x$. Analysis $\ell_1$-minimization tries to recover the signal by computing the minimizer of
\begin{equation}\label{eqProblemP1Noise}
\underset{z\in\RR^d}\min\;\norm{\Omega z}_1\;\;\mbox{subject to}\;\;\; \norm{Mz-y}_2\leq\eta.
\end{equation}

In this paper we consider two prominent examples of the analysis operator. The method of total variation corresponds to the program (\ref{eqProblemP1Noise}), when $\Omega$ is a difference operator. In the one-dimensional case it is defined by the matrix
\begin{equation}\label{eq:OneDimDifferenceOperator}
\Omega= \left(\begin{matrix} -1 & 1 & 0 & \cdots & 0\\
0 & -1 & 1 & \cdots & 0 \\
\vdots & & \ddots & \ddots & \vdots \\
0 & \cdots & 0 &  -1 & 1  \end{matrix}\right).
\end{equation}
In this setting (\ref{eqProblemP1Noise}) promotes piecewise constant signals with sparse gradient.
Another important example of the analysis operator appears when the rows $\omega_j$ of $\Omega$ form a frame, i.e., if there exist constants $0 < A\leq B < \infty$ such that
\begin{equation}\label{eq:Frame}
A \|x\|_2^2 \leq \| \Omega x \|_2^2 = \sum_{j=1}^p |\langle \omega_j,x\rangle|^2 \leq B \|x\|_2^2.
\end{equation}

We are interested in the minimal number of measurements in terms of the sparsity (or cosparsity) required to recover a cosparse vector from its measurements $y$ in \eqref{eqMeasurements} when
the matrix $M$ is a Gaussian random matrix, i.e., its entries are independent standard normal distributed random variables.
We rely on a recent result of Foygel and Mackey in \cite{foma14} which is in spirit of the geometric approach of \cite{chparewi12} in order to provide such a bound on the number of Gaussian
measurements needed to recover $x$ via (\ref{eqProblemP1Noise}), when $\Omega$ satisfies either (\ref{eq:OneDimDifferenceOperator}) or (\ref{eq:Frame}).


\subsection{Main Results}

We first provide a general bound for the required number $m$ of Gaussian measurements in order to (stably) recover an analysis-sparse signal $x \in \RR^d$
via analysis $\ell_1$-minimization, see Theorem~\ref{thNumberOfMeasurementsForDecomposableNorms}.
An important feature of the result is that $m$ is always (essentially) less than the ambient dimension $d$. Based on this, we show that
a signal $x \in \RR^d$
which is $s$-sparse with respect to the difference operator $\Omega$ in \eqref{eq:OneDimDifferenceOperator} (that is, whose gradient is $s$-sparse) can be recovered with high probability if ``roughly'', i.e., ignoring terms
of lower order
\begin{equation}\label{m:TV:rough}
m > d\left( 1- \frac{1}{\pi}(1- \frac{s+1}{d})^2\right),
\end{equation}
see Theorem~\ref{thNumberOfMeasurementsForL1NormDifOperator}. In \eqref{m:TV:rough}, the number of measurements is clearly less than $d$. Note that the usual bound in compressive sensing require
\begin{equation}\label{m:standard}
m > c s \log(d/s)
\end{equation}
for recovery of $s$-sparse vectors via $\ell_1$-minimization from Gaussian random measurements \cite[Chapter 9]{fora13}, \cite{chparewi12}.
One realizes a structural difference of this bound to the new one stated above.
In fact, even for sparsity $s=1$, \eqref{m:TV:rough} roughly requires $m > (1-\pi^{-1}) d \approx  0.6817 \, d$, while \eqref{m:standard} requires $m \geq c \log(d)$, which is of much smaller order. It should be pointed out that a bound of the form (\ref{m:standard}) is so far not available for TV-minimization with $1$-dimensional difference operator. For this particular case a bound which resembles (\ref{m:standard}) is obtained in \cite{CaiXu} and it requires 
\begin{equation}\label{m:CaiXuResult}
m\geq c(sd)^{1/2} \log(d).
\end{equation}
However, if $s$ is proportional to $d$, i.e., $s = \alpha d$ for some $\alpha \in (0,1)$ then \eqref{m:TV:rough} requires $m \geq c_\alpha d$ with $c_\alpha = 1-\frac{1}{\pi}(1-\alpha)^2$ which is always less than $1$, while
\eqref{m:standard} and \eqref{m:CaiXuResult} give $m \geq c'_\alpha d$ with $c'_\alpha = c \alpha \ln( \alpha^{-1})$ which may become larger than $1$ because the available estimates for $c$ in \eqref{m:standard} are larger than $2$ \cite{fora13}.

As the second case, we consider $\Omega \in \RR^{p \times d}$ to be a frame with frame bounds $A,B > 0$ in \eqref{eq:Frame}. Theorem~\ref{thNumberOfMeasurementsForL1NormFrameOperator}
below shows that a signal $x \in \RR^d$ which is $\ell$-cosparse
with cosupport $\Lambda$, i.e., $\supp \Omega x = \Lambda^c$, can be recovered from $m$ Gaussian measurements via analysis $\ell_1$-minimization with high probability if, roughly speaking,
\begin{equation}\label{m:frame:rough}
m > d - \frac{2}{\pi p B}\left(\sum_{\ell \in \Lambda} \|\omega_\ell\|_2\right)^2.
\end{equation}
According to \eqref{m:frame:rough},
the number of measurements is always less than $d$, and even though there is no direct dependence on lower frame bound $A$, it is still independent of the scaling of $\Omega$. The bound derived in \cite{kara13} for analysis sparse recovery with respect to a frame roughly requires
\begin{equation}\label{m:prev:frame}
m \geq \frac{2B}{A} s \ln(ep/s),
\end{equation}
see also \cite{kara14}.
In order to place the bound \eqref{m:frame:rough} of Theorem~\ref{thNumberOfMeasurementsForL1NormFrameOperator} for the frame case into context, we recall some facts on the analysis sparsity model.
Consider the generic case, that the rows of the analysis operator $\Omega$, that is, the frame elements $\omega_j$, $j = 1,\hdots,p$, are in general linear position so that a subcollection of $d$ frame elements
is always linearly independent. Then any subspace $W_\Lambda = \operatorname{span} \{ \omega_j, j \in \Lambda \}^\perp$ with $\#\Lambda = \ell$, of $\ell$-cosparse vectors,
where ${}^\perp$ denotes the orthogonal complement, has dimension $d-\ell$, which means that the smallest value that the sparsity parameter $s = p- \ell$ can take for a nontrivial vector $x$
is $p-d+1$, see also \cite{daelgrna11,kara14}. Therefore, if $p = \kappa d$ for some $\kappa > 1$, then the sparsity is always proportional to the dimension $d$, i.e., $s \sim \alpha d$
which means that the bound \eqref{m:frame:rough} reads $m\geq c_{\kappa,\alpha}d$ with $c_{\kappa, \alpha}=1-\frac{2(\kappa-\alpha)^2}{\pi\kappa B}$, assuming that the frame is normalized. At the same time the $\log$ factor in \eqref{m:standard} becomes constant for $s$ proportional to $d$ and if $\kappa > 2$, then the sparsity is always at least as large as the ambient dimension, so that
\eqref{m:standard} becomes a trivial bound. Similar considerations
apply to the previous bound in \cite{kara13} on analysis-sparse recovery with respect to frames.
In contrast, let us consider the new bound \eqref{m:frame:rough} for a tight unit norm frame, i.e., $\|\omega_j\|_2=1$ and $A=B = p/d$.
For cosparsity $\# \Omega = \ell = p-s$, \eqref{m:frame:rough} yields then
\[
m > d - \frac{2d}{\pi p^2}(p-s)^2 = \left(1- \frac{2(p-s)^2}{\pi p^2}\right) d.
\]
For sparsity proportional to $p$, the number of measurements clearly scales like $\beta d$ where $\beta$ is always strictly smaller than $1$,
which shows that our new bound is better in the natural regime $p = \kappa d$ with
$\kappa \geq 2$, say. In contrast, the previous bound \eqref{m:prev:frame} requires $m$ to be larger than $d$ for a frame in general position with $p = \kappa d$, see Figure \ref{fig:ComparisonOfBounds}.
\begin{figure}
\centering
\includegraphics[scale=0.5]{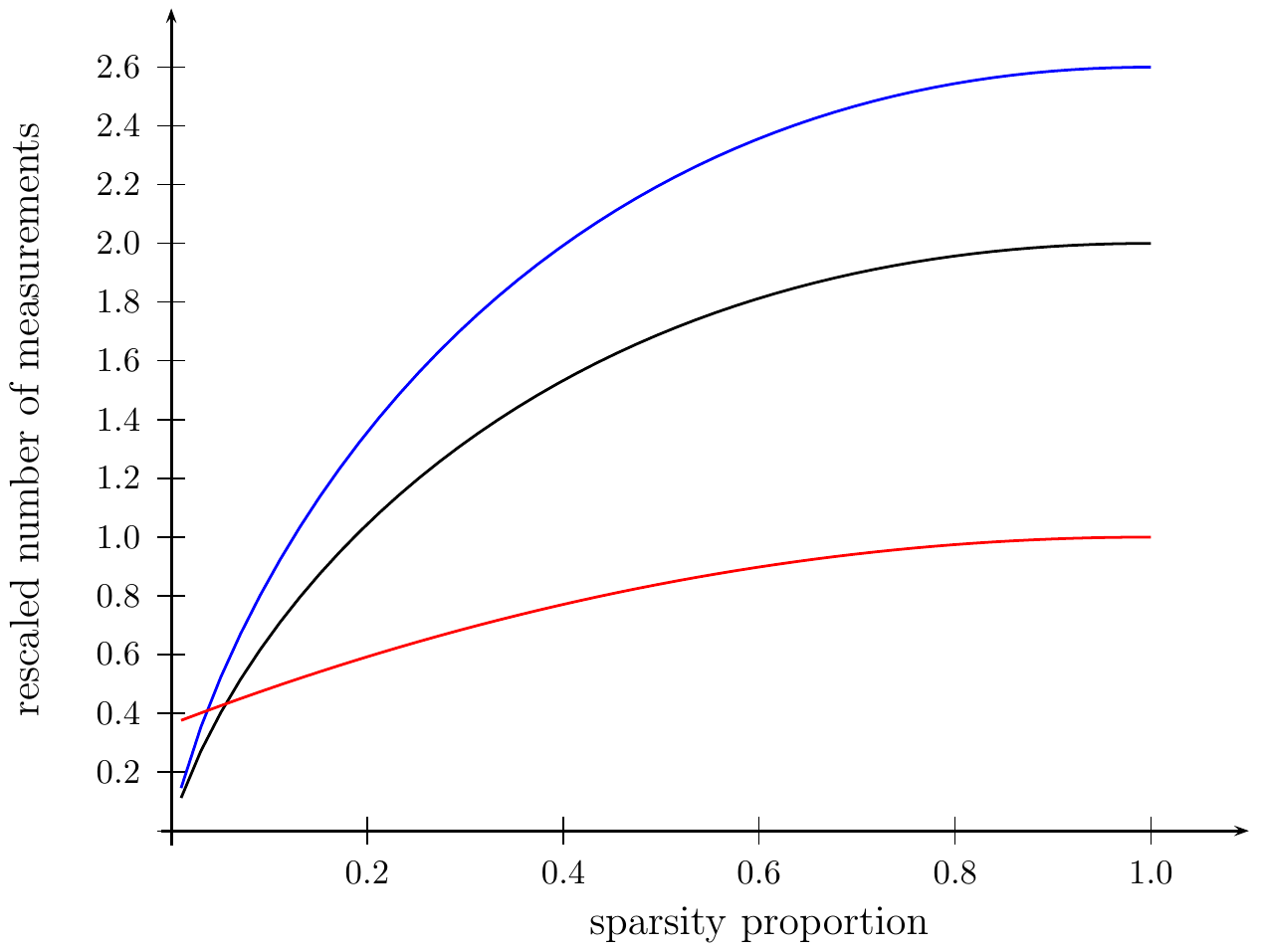}
\caption{Comparison of the standard recovery bounds and the new one for the setting of a normalized tight frame with frame bounds $A=B=\frac{p}{d}$. The red line corresponds to the new bound on the rescaled number of measurements $\frac{m}{d}$ as a function of sparsity fraction $\frac{s}{p}$. The black and blue lines represent the bound given by (\ref{m:prev:frame}) for $p=d$ and $p=1.3d$ respectively.}
\label{fig:ComparisonOfBounds}
\end{figure}


\subsection{Notation}
We denote the rows of $\Omega\in\RR^{p\times d}$ by $\omega_j$, $j=1,\dots,p$. We use $\Omega_{\Lambda}$ to refer to a submatrix of $\Omega$ with the rows indexed by $\Lambda$; $\Omega^T$ is the transpose of $\Omega$; $\norm{\Omega\cdot}_1$ denotes a semi-norm generated by the transform $\Omega$; $\alpha_{\Lambda}$ stands for the vector whose entries indexed by $\Lambda$ coincide with the entries of $\alpha$ and the rest are filled by zeros.
The sign of a real number $r\neq 0$ is $\sgn(r)=\frac{r}{\abs{r}}$.
For a vector $\alpha\in\RR^p$ we define its sign vector $\sgn(\alpha)\in\RR^p$ by
\[
(\sgn(\alpha))_i=\left\{\begin{array}{cl}
\frac{\alpha_i}{\abs{\alpha_i}}, & \mbox{if }\,\alpha_i\neq 0,\\
0, & \mbox{otherwise}.
\end{array}\right.
\]
We write $S^{d-1}$ for the unit sphere in $\RR^d$.


\section{Recovery from Gaussian measurements}

\subsection{Theoretical guarantees}
Our analysis is based on work of Chandrasekaran et al.\ \cite{chparewi12}, where sufficient conditions for robust recovery make use of the tangent cone (also called descent cone) of
a convex function (e.g., the $\ell_1$-norm) at the signal to be recovered, see also \cite{tr14}.
In the case of analysis $\ell_1$-minimization these are formulated explicitly in \cite{kara13}. For a fixed $x\in\RR^d$ we define a tangent cone as
\[
T(x)=\cone\{z-x:z\in\RR^d,\;\norm{\Omega z}_1\leq\norm{\Omega x}_{1}\},
\]
where the notation ``cone'' stands for the conic hull of the indicated set. The set $T(x)$ consists of the directions from $x$, which do not increase the value of $\norm{\Omega x}_{1}$.
\begin{theorem}[\cite{kara13}]\label{thm:rec}
Let $x\in\RR^d$. If
\begin{equation}\label{eqTangentConeWithNoise}
\underset{\substack{v\in T(x)\\ \norm{v}_2\leq 1}}\inf \norm{Mv}_2\geq\tau,
\end{equation}
for some $\tau>0$, then a minimizer $\hat x$ of (\ref{eqProblemP1Noise}) satisfies
\[
\norm{x-\hat x}_2\leq \frac{2\eta}{\tau}.
\]
\end{theorem}

\subsection{Gordon's escape through a mesh theorem}
According to (\ref{eqTangentConeWithNoise}), successful recovery of a signal is achieved, when the minimal gain of the measurement matrix over the tangent cone is greater than some positive constant. For Gaussian matrices the probability of this event can be estimated by Gordon's escape through a mesh theorem \cite{go88-2}, see also \cite[Theorem 9.21]{fora13}. In order to present Gordon's result formally, we introduce some notation.

Let $g\in\RR^m$ be a standard Gaussian random vector. Then
\[
E_m:=\mean\norm{g}_2=\sqrt{2}\;\frac{\Gamma\brac{(m+1)/2}}{\Gamma\brac{m/2}},
\]
where $\Gamma$ is the standard Gamma function, satisfies
\[
\frac{m}{\sqrt{m+1}}\leq E_m\leq\sqrt{m}.
\]
For a set $T\subset\RR^d$ we define its Gaussian width by
\[
\ell(T):=\mean\underset{x\in T}\sup\;\abrac{x,g},
\]
where $g\in\RR^d$ is a standard Gaussian random vector.
\begin{theorem}[Gordon's escape through a mesh]\label{thGordonsEscapeThroughTheMesh}
Let $M\in\RR^{m\times d}$ be a Gaussian random matrix and $T$ be a subset of the unit sphere $S^{d-1}=\{x\in\RR^d: \norm{x}_2=1\}$. Then, for $t>0$, 
\begin{equation}\label{eqGordonsEscapeThroughTheMesh}
\PP\brac{\underset{x\in T}\inf\norm{Mx}_2> E_m-\ell(T)-t}\geq 1-e^{-\frac{t^2}{2}}.
\end{equation}
\end{theorem}
Thus, to provide a bound on the number of Gaussian measurements, we estimate the Gaussian width of the tangent cone $T(x)$ intersected with the unit sphere.

\subsection{Estimates for the Gaussian width}

A basic technique to estimate the Gaussian width of the tangent cone was developed in the work of Stojnic \cite{st09-6} and was later refined in a series of papers \cite{amlomctr13,chparewi12}, see also \cite{tr14}.
The result states that the Gaussian width of the tangent cone can be bounded by the Euclidean distance of the standard normal vector to the scaled cone generated by the subdifferential.

Recall that the Euclidean distance to a set $T \subset \RR^d$ is the function defined by
\[
\dist(x,T):=\inf\fbrac{\norm{x-u}_2:u\in T}.
\]
The subdifferential of a convex function $f: \RR^d \to \RR$ at a point $x \in \RR^d$ is the set
\[
\partial f(x) = \{ v \in \RR^d : f(z) \geq f(x) + \langle z-x,v \rangle \mbox{ for all } z \in \RR^d \}.
\]
It is shown in \cite{st09-6,chparewi12} that
\begin{equation}\label{eqGaussianWidthEuclideanDistance}
\ell(T(x)\cap\sph^{d-1})^2\leq\underset{t\geq 0}\inf\mean\sbrac{\dist(g,t\cdot\partial\norm{\Omega\cdot}_{1}(x))}^2,
\end{equation}
where $g\in\RR^d$ is the standard normal vector.

To provide a bound on the expected squared distance from (\ref{eqGaussianWidthEuclideanDistance}) we generalize Proposition 3 in \cite{foma14} valid for $\Omega$ being a basis to the following result.
\begin{lemma}
Let $\Omega\in\RR^{p\times d}$ be an analysis operator, $g\in\RR^d$ be a standard normal random vector and $x\in\RR^d$. Then
\begin{equation}\label{eqEstimateGWFoygel}
\underset{t\geq 0}\inf\mean\sbrac{\dist(g,t\cdot\partial\norm{\Omega\cdot}_{1}(x))}^2\leq d-\frac{\sbrac{\ell(\partial\norm{\Omega\cdot}_{1}(x))}^2}{\underset{\norm{z}_2\leq 1}\max\norm{\Omega z}_{1}^2}.
\end{equation}
\end{lemma}
\begin{proof}
Since $\partial \norm{\Omega\cdot}_1(x)$ is a compact set, there exists $w_0\in \partial \norm{\Omega\cdot}_1(x)$ such that
\[
\abrac{g,w_0}=\underset{w\in\partial \norm{\Omega\cdot}_1(x)}\max\abrac{g,w}.
\] 
Then 
\[
\dist(g,t\cdot\partial\norm{\Omega\cdot}_{1}(x))^2\leq \norm{g-tw_0}_2^2=\norm{g}_2^2-2t\abrac{g,w_0}+t^2\norm{w_0}_2^2.
\]
Moreover,
\[
\partial\norm{\Omega\cdot}_{1}(x)=\Omega^T\brac{\partial\norm{\cdot}_{1}(\Omega x)}
\]
and
\[
\partial\norm{\cdot}_1(\Omega x)=\fbrac{\alpha\in\RR^p:\alpha_{\Lambda^c}=\sgn(\Omega x),\;\;\norm{\alpha_{\Lambda}}_{\infty}\leq 1},
\]
which implies that $w_0=\Omega^T\alpha_0$ for some $\alpha_0\in\partial\norm{\cdot}_1(\Omega x)$. Thus
\[
\dist(g,t\cdot\partial\norm{\Omega\cdot}_{1}(x))^2\leq\norm{g}_2^2-2t\underset{w\in\partial \norm{\Omega\cdot}_1(x)}\max\abrac{g,w}+t^2\norm{\Omega^T\alpha_0}_2^2.
\]
Due to duality and the fact that $\norm{\alpha_0}_{\infty}\leq 1$ we obtain
\[
\norm{\Omega^T\alpha_0}_2=\underset{\norm{z}_2\leq 1}\max\abrac{z,\Omega^T\alpha_0}=\underset{\norm{z}_2\leq 1}\max\abrac{\Omega z,\alpha_0}\leq\underset{\norm{z}_2\leq 1}\max\norm{\Omega z}_1\norm{\alpha_0}_{\infty}\leq\underset{\norm{z}_2\leq 1}\max\norm{\Omega z}_1.
\]
So altogether 
\[
\dist(g,t\cdot\partial\norm{\Omega\cdot}_{1}(x))^2\leq \norm{g}_2^2-2t\underset{w\in\partial \norm{\Omega\cdot}_1(x)}\max\abrac{g,w}+t^2\underset{\norm{z}_2\leq 1}\max\norm{\Omega z}_1^2
\]
and by taking the expectation of both sides we obtain
\[
\mean\dist(g,t\cdot\partial\norm{\Omega\cdot}_{1}(x))^2\leq d-2t\ell(\partial\norm{\Omega\cdot}_1(x))+t^2\underset{\norm{z}_2\leq 1}\max\norm{\Omega z}_1^2.
\]
Setting $t=\frac{\ell(\partial\norm{\Omega\cdot}_1(x))}{\underset{\norm{z_2}\leq 1}\max\norm{\Omega z}_1^2}$ yields
\[
\underset{t\geq 0}\inf\mean\sbrac{\dist(g,t\cdot\partial\norm{\Omega\cdot}_{1}(x))}^2\leq d-\frac{\sbrac{\ell(\partial\norm{\Omega\cdot}_{1}(x))}^2}{\underset{\norm{z}_2\leq 1}\max\norm{\Omega z}_{1}^2}.
\]
\end{proof}


\subsection{Number of Gaussian measurements}
Gordon's escape through a mesh, Theorem \ref{thGordonsEscapeThroughTheMesh}, together with the estimates (\ref{eqGaussianWidthEuclideanDistance}) and (\ref{eqEstimateGWFoygel}) leads to the next result.
%
\begin{theorem}\label{thNumberOfMeasurementsForDecomposableNorms}
Let $\Omega\in\RR^{p\times d}$ be an analysis operator and $x$ be cosparse with cosupport $\Lambda$. For a random draw $M\in\RR^{m\times d}$ of a Gaussian matrix, let noisy measurements $y=Mx+w$ be given with $\norm{w}_2\leq\eta$ and $0<\eps<1$. If
\begin{equation}\label{eqNumberOfMeasurementsByGW}
\frac{m^2}{m+1}\geq \brac{\sqrt{d-\frac{\brac{\mean\norm{\Omega_{\Lambda}g}_{1}}^2}{\underset{\norm{z}_2\leq 1}\max\norm{\Omega z}_{1}^2}}+\sqrt{2\ln(\eps^{-1})}+\tau}^2,
\end{equation}
then with probability at least $1-\eps$, any minimizer $\hat x$ of (\ref{eqProblemP1Noise}) satisfies
\[
\norm{x-\hat x}_2\leq\frac{2\eta}{\tau}.
\]
\end{theorem}
\begin{proof}
Estimates (\ref{eqGaussianWidthEuclideanDistance}) and (\ref{eqEstimateGWFoygel}) imply that
\begin{equation}\label{eq:PreEstimateWithGW}
E_m-\ell\brac{T(x)\cap \sph^{d-1}}-t\geq E_m-\sqrt{d-\frac{\sbrac{\ell(\partial\norm{\Omega\cdot}_{1}(x))}^2}{\underset{\norm{z}_2\leq 1}\max\norm{\Omega z}_{1}^2}}-t.
\end{equation}
For any $z\in\partial\norm{\Omega\cdot}_{1}(x)$ there is $\alpha\in\RR^p$ with $\norm{\alpha_{\Lambda}}_{\infty}\leq 1$ such that
\[
z=\Omega^T\alpha=\Omega^T\sgn(\Omega x)+\Omega^T\alpha_{\Lambda}.
\]
Therefore,
\[
\ell(\partial\norm{\Omega\cdot}_{1}(x))=\mean\underset{z\in\partial\norm{\Omega\cdot}_{1}(x)}\max\abrac{g,z}=\mean\abrac{g,\Omega^T\sgn(\Omega x)}+\mean\underset{\norm{\alpha_{\Lambda}}_{1}\leq 1}\max\abrac{g,\Omega^T\alpha_{\Lambda}}
\]
\[
=\mean\underset{\norm{\alpha_{\Lambda}}_{1}\leq 1}\max\abrac{(\Omega g)_{\Lambda},\alpha_{\Lambda}}=\mean\norm{\Omega_{\Lambda}g}_{1}.
\]
Plugging this into (\ref{eq:PreEstimateWithGW}) gives
\[
E_m-\ell\brac{T(x)\cap \sph^{d-1}}-t\geq E_m-\sqrt{d-\frac{\brac{\mean\norm{\Omega_{\Lambda}g}_{1}}^2}{\underset{\norm{z}_2\leq 1}\max\norm{\Omega z}_{1}^2}}-t.
\]
Setting $t=\sqrt{2\ln(\eps^{-1})}$, the choice of $m$ in (\ref{eqNumberOfMeasurementsByGW}) guarantees that
\[
E_m-\ell\brac{T(x)\cap \sph^{d-1}}-t\geq\frac{m}{\sqrt{m+1}}-\sqrt{d-\frac{\brac{\mean\norm{\Omega_{\Lambda}g}_{1}}^2}{\underset{\norm{z}_2\leq 1}\max\norm{\Omega z}_{1}^2}}-\sqrt{2\ln(\eps^{-1})}\geq\tau.
\]
The monotonicity of probability together with Theorems~\ref{thm:rec} and \ref{thGordonsEscapeThroughTheMesh} yields
\[
\begin{aligned}
\PP\brac{\underset{{x}\in T(x)\cap\sph^{d-1}}\inf\norm{Mx}_2\geq\tau}&\geq\PP\brac{\underset{x\in T(x)\cap\sph^{d-1}}\inf\norm{Mx}_2\geq E_m-\ell\brac{T(x)\cap\sph^{d-1}}-t}\\
&\geq 1-\eps.
\end{aligned}
\]
This concludes the proof.
\end{proof}

\section{Explicit choice of the analysis operator}

Theorem \ref{thNumberOfMeasurementsForDecomposableNorms} can be further refined for special choices of the operator $\Omega$. We start with the one-dimensional difference operator.

\begin{theorem}\label{thNumberOfMeasurementsForL1NormDifOperator}
Let $\Omega:\RR^d\to\RR^{d-1}$ be a one-dimensional difference operator. Let $x\in\RR^d$ be $\ell$-cosparse with respect to $\Omega$ and $s=d-1-\ell$.
For a random draw $M\in\RR^{m\times d}$ of a Gaussian matrix, let noisy measurements $y=Mx+w$ be given with $\norm{w}_2\leq\eta$ and $0<\eps<1$. If
\begin{equation}\label{eqNumberOfMeasurementsForL1NormDifOperator}
\frac{m^2}{m+1}\geq \brac{\sqrt{d\brac{1-\frac{1}{\pi}\brac{1-\frac{s+1}{d}}^2}}+\sqrt{2\ln(\eps^{-1})}+\tau}^2,
\end{equation}
then with probability at least $1-\eps$, any minimizer $\hat x$ of (\ref{eqProblemP1Noise}) satisfies
\[
\norm{x-\hat x}_2\leq\frac{2\eta}{\tau}.
\]
\end{theorem}
\begin{proof}
According to the definition of $\Omega$, for any $z\in\RR^d$ 
\begin{equation}\label{eq:PropertiesOfDiffOperator}
\norm{\Omega z}_1=\sum\limits_{i=1}^{d-1}\abs{z_{i+1}-z_i}\leq 2\norm{z}_1\leq 2\sqrt d\norm{z}_2.
\end{equation}
Hence, $\underset{\norm{z}_2\leq 1}\max\norm{\Omega z}_1^2\leq 4d$.
The rows $\omega_j$ of $\Omega\in\RR^{(d-1)\times d}$ satisfy $\norm{\omega_i}_2=\sqrt 2$ and the properties of the Gaussian distribution imply
\begin{equation}\label{eq:PropertiesGaussianDistribution}
\mean\norm{\Omega_{\Lambda}g}_1=\mean\sum_{i\in\Lambda}\abs{\abrac{g,\omega_i}}=\sqrt{\frac{2}{\pi}}\sum\limits_{i\in\Lambda}\norm{\omega_i}_2=\frac{2}{\sqrt\pi}(d-1-s).
\end{equation}
Plugging (\ref{eq:PropertiesOfDiffOperator}) and (\ref{eq:PropertiesGaussianDistribution}) into (\ref{eqNumberOfMeasurementsByGW}) implies the desired result.
\end{proof}
A comparison of the theoretical bound to the bound obtained from the numerical experiments is shown on Figure \ref{fig:DiffOperatorNumericalExperiments}. We considered signals in $\RR^{200}$. We ran the algorithm and counted the number of times the signal was recovered correctly out of 200 trials. A reconstruction error of $10^{-5}$ was considered as a successful recovery. Each pixel intensity represents
the ratio of the signals recovered perfectly with black being 100\% success. 
\begin{figure}
\centering
\includegraphics[scale=0.3]{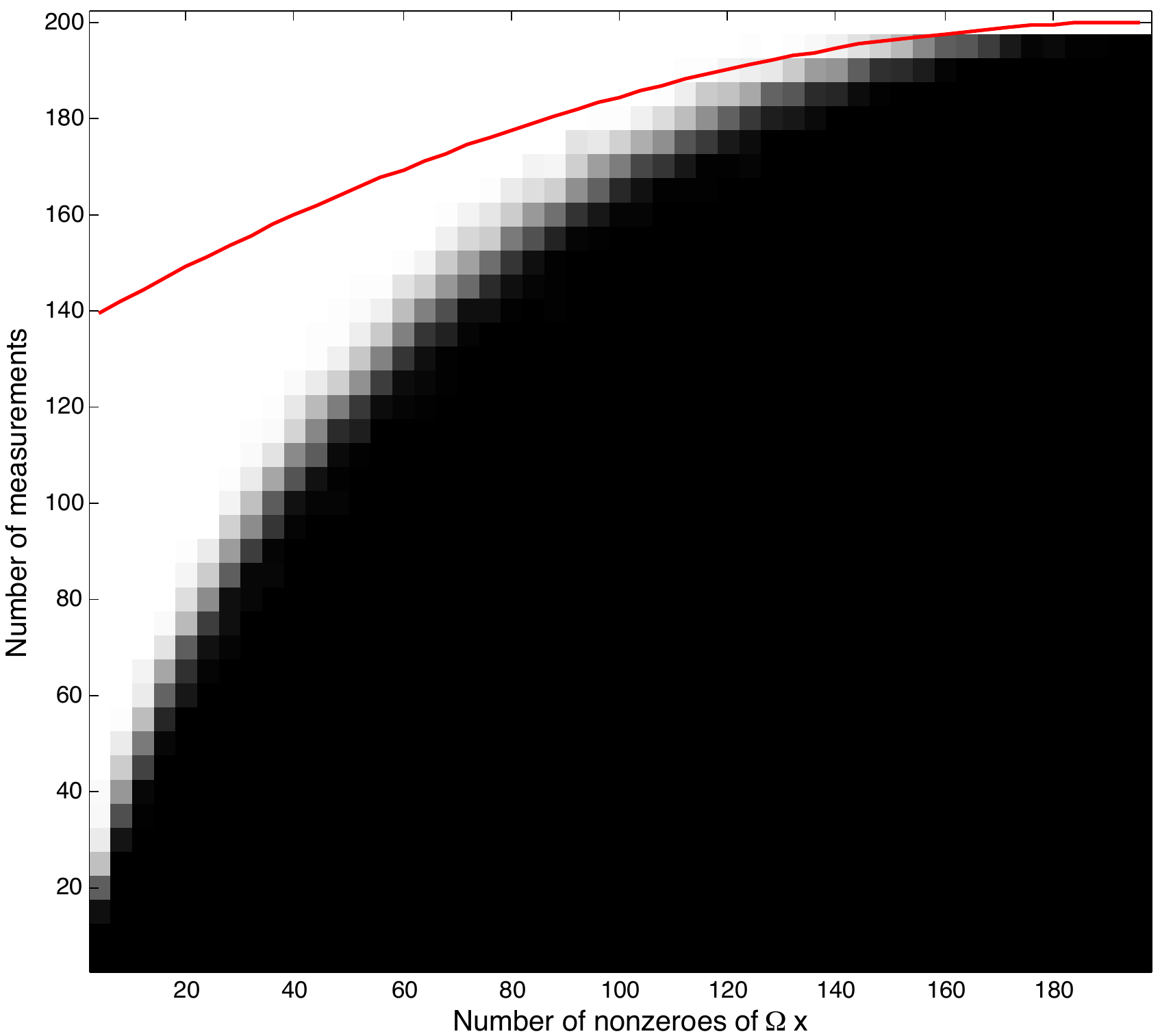}
\caption{The red line corresponds to the theoretical bound omitting the terms of lower order.}
\label{fig:DiffOperatorNumericalExperiments}
\end{figure}
We would like to point out that, even though the theoretical guarantee requires around 90\% of the samples to be taken, when 80 out of 199 entries of the gradient are non-zero, minimizing the TV-norm would still provide better performance than simply applying Tikhonov regularization with 90\% samples. We present a comparison of the results of the two reconstruction algorithms in Figure \ref{fig:ComparisonTVandLS}.
\begin{figure}
\minipage{0.32\textwidth}
\includegraphics[width=\linewidth]{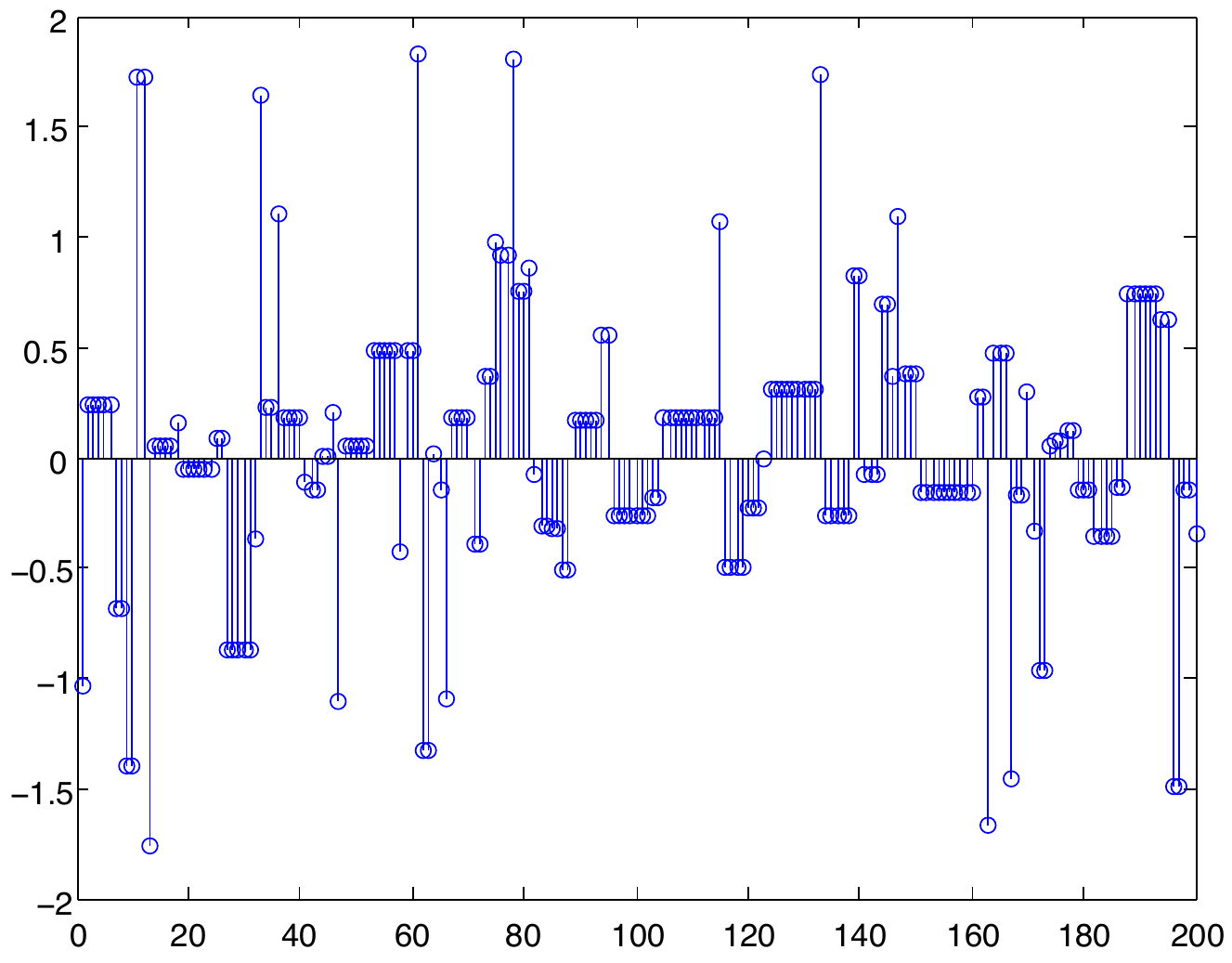}
\endminipage\hfill
\minipage{0.32\textwidth}
  \includegraphics[width=\linewidth]{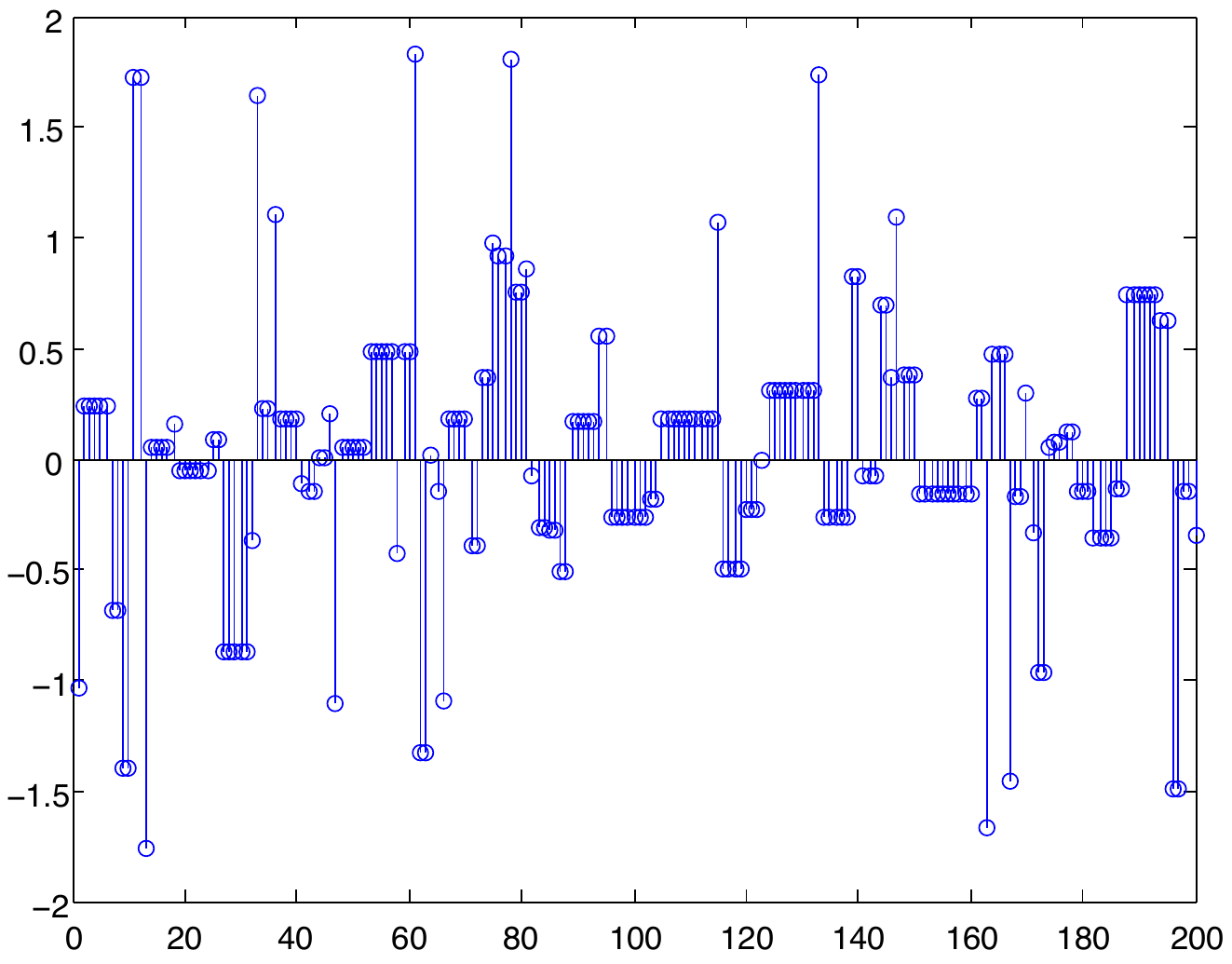}
\endminipage\hfill
\minipage{0.32\textwidth}%
  \includegraphics[width=\linewidth]{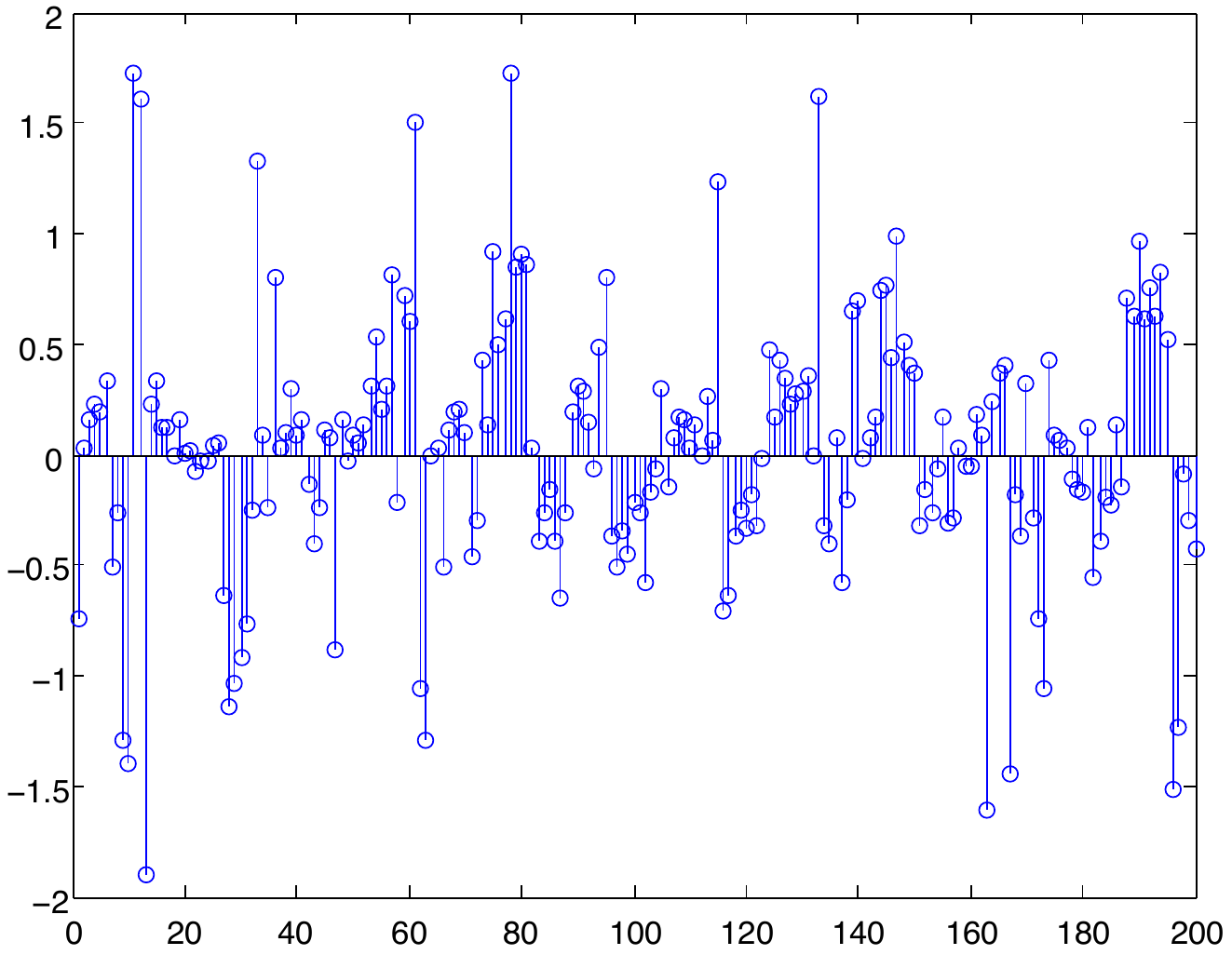}
\endminipage
 \caption{The first plot depicts an initial signal in $\RR^{200}$, whose gradient has 80 non-zero entries. The second and third one correspond to the reconstruction from 180 measurements via TV-minimization (exact) and Tikhonov regularization respectively.}\label{fig:ComparisonTVandLS}
\end{figure}

Theorem \ref{thNumberOfMeasurementsForL1NormDifOperator} can be extended to two dimensions. Let $X\in\RR^{d\times d}$. The two-dimensional difference operator collects all vertical and horizontal derivatives of $X$ into a single vector. If we concatenate the columns of $X$ into the vector $x\in\RR^{d^2}$, then we can represent this operator by the matrix $\Omega\in\RR^{2d(d-1)\times d^2}$, whose action is given by
\[
\Omega x=\brac{X_{21}-X_{11},\ldots,X_{dd}-X_{d-1d},X_{12}-X_{11},\ldots,X_{dd}-X_{dd-1}}^T.
\]
Each row $\omega_i$ of $\Omega$ has exactly two non-zero entries with values $-1$ and $1$ at the proper locations.
\begin{theorem}\label{th2DNumberOfMeasurementsForDifOperator}
Let $\Omega\in\RR^{2d(d-1)\times d^2}$ define a two-dimensional difference operator. Let $x\in\RR^{d^2}$ be $\ell$-cosparse with respect to $\Omega$ and $s=2d(d-1)-\ell$.
For a random draw $M\in\RR^{m\times d^2}$ of a Gaussian matrix, let noisy measurements $y=Mx+w$ be given with $\norm{w}_2\leq\eta$ and $0<\eps<1$. If
\begin{equation}\label{eq2DNumberOfMeasurementsForDifOperator}
\frac{m^2}{m+1}\geq \brac{\sqrt{d^2\brac{1-\frac{1}{\pi}\brac{1-\frac{1}{d}-\frac{s}{2d^2}}^2}}+\sqrt{2\ln(\eps^{-1})}+\tau}^2,
\end{equation}
then with probability at least $1-\eps$, any minimizer $\hat x$ of (\ref{eqProblemP1Noise}) satisfies
\[
\norm{x-\hat x}_2\leq\frac{2\eta}{\tau}.
\]
\end{theorem}
\begin{proof}
Let $z\in\RR^{d^2}$. Then
\[
\norm{\Omega z}_1=\sum_{j=1}^d\sum_{i=1}^{d-1}\abs{z_{i+1j}-z_{ij}}+\sum_{i=1}^d\sum_{j=1}^{d-1}\abs{z_{ij+1}-z_{ij}}\leq 4\norm{z}_1\leq 4d\norm{z}_2
\]
and it follows that
\[
\underset{\norm{z}_2\leq 1}\max\norm{\Omega z}_1^2\leq 16 d^2.
\]
As in the one-dimensional case,
\[
\mean\norm{\Omega_{\Lambda}g}_1=\sqrt{\frac{2}{\pi}}\sum_{i\in\Lambda}\norm{\omega_i}_2=\frac{2}{\sqrt{\pi}}(2d(d-1)-s).
\]
As the final step we apply formula (\ref{eqNumberOfMeasurementsByGW}).
\end{proof}

For an $\Omega$ being a frame, we obtain the following bound on the required number of measurements.
\begin{theorem}\label{thNumberOfMeasurementsForL1NormFrameOperator}
Let $\Omega:\RR^d\to\RR^{p}$ be a frame with an upper frame bound $B$. Let $x\in\RR^d$ be cosparse with cosupport $\Lambda$. For a random draw $M\in\RR^{m\times d}$ of a Gaussian matrix, let noisy measurements $y=Mx+w$ be given with $\norm{w}_2\leq\eta$ and $0<\eps<1$. If
\begin{equation}\label{eqNumberOfMeasurementsForL1NormFrameOperator}
\frac{m^2}{m+1}\geq \brac{\sqrt{d-\frac{2}{\pi}\frac{\brac{\sum_{i\in\Lambda}\norm{\omega_i}_2}^2}{pB}}+\sqrt{2\ln(\eps^{-1})}+\tau}^2,
\end{equation}
then with probability at least $1-\eps$, any minimizer $\hat x$ of (\ref{eqProblemP1Noise}) satisfies
\[
\norm{x-\hat x}_2\leq\frac{2\eta}{\tau}.
\]
\end{theorem}
\begin{proof}
The only difference to the proof of Theorem \ref{thNumberOfMeasurementsForL1NormDifOperator} is the following estimate, which is due to the Cauchy inequality and the fact that $\Omega$ is a frame:
\[
\underset{\norm{z}_2\leq 1}\max\norm{\Omega z}_1^2\leq p\underset{\norm{z}_2\leq 1}\max\norm{\Omega z}_2^2\leq p\underset{\norm{z}_2\leq 1}\max B\norm{z}_2^2\leq pB.
\]
\end{proof}
The bound on the number of measurements (\ref{eqNumberOfMeasurementsForL1NormFrameOperator}) does not have an explicit dependence on the ratio of the frame bounds.
So it can not be directly compared to the results provided in \cite{kara13}, see also \eqref{m:prev:frame}.
However, an important observation is that the right hand side in (\ref{eqNumberOfMeasurementsForL1NormFrameOperator}) is strictly less than $d$ for any $p$ and any number of elements in the cosupport of  the signal
(provided $\varepsilon^{-1}$ and $\tau$ are not too large).


\section{Conclusions}

We have presented results on the nonuniform recovery from Gaussian random measurements of analysis-sparse signals with respect to the one- and two-dimensional difference operators or with respect to a frame. The derived bound on required measurements is always smaller than the ambient dimension of a signal and it is particularly suitable for the case when the sparsity of the analysis representation of the signal is not very small. 


\section{Acknowledgements}
M.~Kabanava and H.~Rauhut acknowledge support by the European Research Council through the grant StG 258926. H. Zhang is supported by China NSF Grants No. 61201328.



\bibliography{TVGaussBib}
\bibliographystyle{abbrv}

\end{document}